\documentclass[conference]{IEEEtran}
\usepackage{ifpdf}
\usepackage{cite}
\ifCLASSINFOpdf
  \usepackage[pdftex]{graphicx}
  \usepackage{pdfpages}
  \graphicspath{{images/}}
\else
\fi
\usepackage{amsmath}
\usepackage{algorithmic}
\usepackage{array}
\usepackage{fixltx2e}
\usepackage{url}
\usepackage{amsmath,epsfig,amssymb,amsbsy,fixmath,mathrsfs,psfrag,epsf,enumerate,bm,color, epstopdf, pdfpages, bm}
\usepackage{graphicx}
\usepackage{amssymb}
\usepackage{graphicx}
\usepackage[caption=false]{subfig}
\usepackage{booktabs}

\usepackage{amsfonts}
\usepackage{bbm}
\usepackage[utf8]{inputenc}

\usepackage{mathtools, cuted}

\usepackage{amsmath, amsthm, amssymb,bm,eufrak} 
\newtheorem{theorem}{Theorem}[section]
\newtheorem{corollary}{Corollary}[theorem]
\newtheorem{lemma}[theorem]{Lemma}

\newtheorem{proposition}{Proposition}

\begin{document}
\title{Performance Analysis of Inband FD-D2D Communications with Imperfect SI Cancellation for Wireless Video Distribution}
\author{\IEEEauthorblockN{Mansour Naslcheraghi${{}^{1,*}}$, \textit{Member, IEEE} \\ Seyed Ali Ghorashi${{}^{1,2}}$, \textit{Senior Member, IEEE}, Mohammad Shikh-Bahaei${{}^{3}}$, \textit{Senior Member, IEEE}}
\IEEEauthorblockA{1. Department of Electrical Eng., Shahid Beheshti University, G.C. Tehran, Iran\\
2. Cyberspace Research Institute, Shahid Beheshti University, G.C. Tehran, Iran\\
3. Centre for Telecommunications Research, King's College London, UK\\
\url{m.naslcheraghi@mail.sbu.ac.ir}, \url{a_ghorashi@sbu.ac.ir}, \url{m.sbahaei@kcl.ac.uk}
}}
\maketitle
\begin{abstract}
 Tremendous growing demand for high data rate services is the main driver for increasing traffic in wireless cellular networks. Device-to-Device (D2D) communications have recently been proposed to offload data via direct communications by bypassing cellular base stations (BSs). Such an offloading schemes increase capacity and reduce end-to-end delay in cellular networks and help to serve the dramatically increasing demand for high data rate. In this paper, we aim to analyze inband full-duplex (FD) D2D performance for the wireless video distribution by considering imperfect self-interference (SI) cancellation. Using tools from stochastic geometry, we analyze outage probability and spectral efficiency. Analytic and simulation results are used to demonstrate achievable gain against its half-duplex (HD) counterpart.
\end{abstract}

\begin{IEEEkeywords}
stochastic geometry, D2D, outage probability, full-duplex, video distribution.
\end{IEEEkeywords}

\IEEEpeerreviewmaketitle
\section{introduction}
Increasing demand for high data rate and live video streaming in cellular networks has attracted researchers' attention to cache-enabled cellular network architectures \cite{Negin_Mag}. These networks exploit D2D communications as a promising technology of 5G heterogeneous networks, for cellular video distribution. In a cellular content delivery network assisted by D2D communications, user devices can capture their desired contents, either via cellular infrastructure or via D2D links from other devices in their vicinity. Recently, several studies in both content placement policies and delivery strategies are conducted to minimize the downloading time, and to maximize the overall network throughput in terms of rate and area spectral efficiency. From the content placement perspective of view, contents can be placed on collaborative nodes (user devices) formerly, either according to a predefined caching policy (reactive caching) \cite{Reactive4}, or more intelligently, according to statistics of user devices' interest (proactive caching) \cite{Proactive1}. Higher layer protocols for resource allocation and scheduling have been studied for LTE and WiFi networks with half duplex transceivers \cite{Mohammad1,Mohammad2,Mohammad3,Mohammad4}. In contrast with device-centric caching policy, cluster-centric caching policy is proposed in \cite{Afshang1} to maximize area spectral efficiency. The main idea in \cite{Afshang1} is to place the content in each cluster in order to maximize the collective performance of all devices. Most of conventional architectures of wireless video distribution schemes in both wireless cellular and D2D opportunities, are based on half-duplex (HD) transmission and, recently full-duplex (FD) capability and its advantages in wireless video distribution with D2D opportunities has been investigated in \cite{MyIET}. It is shown that FD capability can promise near double spectral efficiency, providing self-interference (SI) is entirely canceled. Recent advances in FD radio design \cite{full-duplex-advances1}, materialized by advanced signal processing
techniques that can suppress SI at the receiver, have enabled simultaneous transmission and reception over the same frequency band, which is called inband FD communication. Our main contribution in this paper is analyzing the performance of the FD-enabled D2D opportunities in a cache-enabled cellular network, by employing stochastic geometry that captures physical layer characteristics of an ultra dense network. Details along with the main contributions are as follows.
\begin{itemize}
	\item We propose new system model based on Poisson point process (PPP) to characterize the performance of the system, in terms of closed form expressions of the outage probability and per link spectral efficiency.  
	\item Despite \cite{MyIET} in which perfect SI cancellation considered, in this paper, we consider imperfect SI cancellation that have major impact on FD radios performance.
	\item Despite \cite{MyICT2017} in which HD/FD collaboration probabilities are derived in a uniformly distributed network, in this paper, we provide new analysis for HD/FD collaboration probabilities that captures the extremely random nature of a stochastic network.  
\end{itemize}
Extensive simulations are used to verify analytic results, to prove achievable gain of FD-enabled opportunities in an ultra dense network against its HD counterpart. The remainder of this paper is organized as follows. In section \ref{System Model Section}, system model is clarified. Section \ref{Collaboration probability section} provides the analysis for HD/FD collaboration probabilities. Section \ref{Performance Analysis Section} provides performance analysis of the system. Section \ref{Results and Discussions Section} provides analytic and simulation results, and finally section \ref{Conclusion Section} concludes the paper. 

\section{System Model}
\label{System Model Section}
\subsection{Network and Channel Model}
We consider an in-band overlay spectrum access strategy in which part of cellular spectrum allocated for D2D communications, and therefore there is no interference between D2D and cellular communications \cite{D2D_survey}. All user devices within cell area have capability of operating in FD mode and are under full control of the BS. We also assume that self-interference (SI) cancellation allows the FD radios to transmit and receive simultaneously over the same time/frequency. Both analog and digital SI cancellation methods can be used to partially cancel the SI. However, in practice, it is difficult or even impossible to cancel the SI, perfectly. The SI in FD nodes is assumed to be canceled imperfectly with residual self-interference-to-power ratio $\beta$, where $0 \le \beta \le 1$. The parameter $\beta$ denotes the amount of SI cancellation. When  $\beta  = 0$, there is perfect SI cancellation, and when $\beta=1$, there is no SI cancellation.

A Poisson point process (PPP) $\Psi$ is considered to model the UEs with the intensity of $\lambda$. For D2D links, standard power-law path-loss model is considered in which the signal power decays at the rate of $r^{-\alpha}$, where $\alpha > 2$ is the path-loss exponent. For FD-D2D links, the channel between FD cooperative nodes are reciprocal and hence the path-loss exponent for both links in FD-D2D mode is the same. Independent Rayleigh fading with unit mean is assumed between any D2D pair. D2D transmitters are transmitting with fixed power $\rho_d$.

\subsection{Caching Model}
\label{Caching Model}
Denote the set of popular video contents as $\bm{V} = \{ {v_1},{v_2},...,{v_m}\}$ with size $m$. Each user has a unique identity number $\omega$. We use Zipf distribution for modeling the popularity of video contents \cite{Zipf} and thus, the popularity of the cached video content $v_\omega$ cached in user $u_i$, is inversely proportional to its rank:

\begin{equation}
\label{Zipf formula}
{f_{i\omega}}(\gamma_r,m) = \left( {i^{ {\gamma _r}}} \sum_{j=1}^{m}j^{-\gamma_r}\right)^{-1} ,\begin{array}{*{20}{c}}
{}&{1 \le i \le m}
\end{array}.
\end{equation}
Zipf exponent $\gamma _r$ is skew exponent that characterizes the distribution by controlling the popularity of contents for a given library size, i.e., $m$. 
Assuming each user have a considerable storage to cache video contents, we provide $optimal$ caching policy for content placement strategy in UEs' storage. 
In this mechanism, contents are placed in users' caches in advance in which each user with a considerable storage capacity can cache a subset of contents $\bm{F_{\ell} \subset V}$ from the library, i.e., $\bm{F_{\ell}} = \{ {f_{\ell1}},{f_{\ell2}},...,{f_{\ell \mathfrak{X}}}\}$, $\mathfrak{X} \le m$ and there is no overlap between users' caches, i.e., ${\bm{F}_p}\mathop  \cap \limits_{p \ne q} {\bm{F}_q} = \phi $. Assuming each user has an identity number, in this mechanism, $h$ top cached at the first user ($\ell=1$), $h$ second cached at the second user ($\ell=2$), etc. 

We assume that a fraction of users with intensity $\mu\lambda$ are making request at the same time. We call $\mu \to 0$ as \textit{Low Load} scenario which means few users making request and call $\mu \to 1$ as \textit{High Load} scenario which means all users demand content at the same time. Each user randomly requests a content from the library randomly, and according to Zipf distribution with exponent $\gamma_r$. Although, the signaling mechanisms do not affect our analysis in this work, however, efficient D2D discovery mechanisms can be adopted to discover D2D pairs in an underlay cellular scenario \cite{MyDiscoveryPaper}. 
According to the information of caches and UEs' requests, from D2D node establishment point of view, there are two different possible configurations for a D2D collaborating node; (i) half-duplex D2D (HD-D2D) mode with probability of $\mathcal{P}_{\textup{HD}}$, in which an arbitrary user node within network can transmit a demanding video content to its respective receiver and (ii) full-duplex D2D (FD-D2D) mode with probability of $\mathcal{P}_{\textup{FD}}$, in which any arbitrary node can operate in FD mode to receive its desired content and simultaneously serve for demanding video content at same time/frequency band (inband FD-D2D). According to Slivnyak's theorem \cite{HaenggiBook}, the statistics observed at any arbitrary random point of PPP is the same as those considered at the origin. Hence, in the analysis sections, we consider a typical receiver at the origin and provide analysis for the typical user. According to this, Since the receivers and their density do not affect our main analysis, the cases of interests, namely, $\mathcal{P}_{\textup{HD}}$ and $\mathcal{P}_{\textup{FD}}$, will be derived just by considering transmitters' intensities across the entire network. 
Now, according to \textit{Thinning Theorem} \cite{HaenggiBook}, the caching nodes storing video contents can be modeled as the independent PPP in two different cases as follows: 
\begin{itemize}
	\item An independent PPP $\Phi_{\textup{HD}}$ with intensity $\mu \mathcal{P}_{\textup{HD}}\lambda$. In this case, all transmitting nodes are operating in HD mode. 
	
	\item An independent PPP $\Phi_{\textup{FD}}$ with intensity $\mu \mathcal{P}_{\textup{FD}}\lambda$. In this case, all transmitting nodes are operating in FD mode.
\end{itemize}

\begin{table}[htb]
	\centering
	\caption {Summary of notation} 
	\centering
	\begin{tabular}{|c|l|}
		\hline
		Notation & Description\\
		\hline \hline
		$\bm{V}$ & Video Content Library\\
		\hline
		$v_i$ & $i$th video content in library\\
		\hline
		$m$ & Library size \\
		\hline
		$\gamma_r$ & Requesting Zipf exponent \\
		\hline
		$\bm{F_\ell}$ & Cached contents in user $u_\ell$\\
		\hline
		$\mathfrak{X}$ & Number of cached contents per user\\
		\hline
		$f_{\ell i}$ & $\ell$th file cached in $i$th user\\
		\hline
		$\mu$ & Fraction of the users ($0 \le \mu \le 1$) \\
		\hline
		$\delta$ & D2D operation mode; $\delta \in \{\textup{HD},\textup{FD}\}$ \\
		\hline
		$\Upsilon^{\delta}$ & SINR in receiver of interest by D2D operation mode $\delta$\\
		\hline
		$\rho_d$ & D2D transmit power\\
		\hline
		$\rho_{r}$ & Received power\\
		\hline
		$\mathcal{I}_{d,\delta}$ & Interference from DUEs in mode $\delta$ on the receiver of interest\\
		\hline
		$\beta$ & SI cancellation factor\\
		\hline
		$\alpha$ & path-loss exponent for D2D links\\
		\hline
		$\lambda$ & UEs density\\
		\hline
		$\lambda_{a}$ & Interfering CUEs density\\
		\hline
		$\bm{\Phi}$ & PPP constituted by the UEs\\
		\hline
		$\bm{\Psi}$ & PPP constituted by the macro BSs\\
		\hline
		$h$ & Rayleigh fading channel gain \\
		\hline
		$\sigma^2$ & Noise power\\
		\hline
		$\mathbbm{1}_{\{.\}}$ & Indicator function \\
		\hline
		$\mathcal{R}_d$ & Distance from typical receiver to its transmitter\\
		\hline
		$\theta_d$ & SINR threshold for D2D communication\\
		\hline
	\end{tabular}
	\label{notations}
\end{table}

\section{Analyzing HD/FD Collaboration Probabilities}
\label{Collaboration probability section}

Denoting $\mathcal{P}(i)$ as the popularity of the cached contents at $u_i$, according to the described caching policy in \ref{Caching Model}, we have
\begin{equation}
\label{popularity opt u_i}
f(i) = \sum_{\omega=1}^{\mathfrak{X}} {f_{i\omega}}(\gamma_r,m), 
\end{equation}
where ${f_{i\omega}}(\gamma_c,m)$ is as in eq  (\ref{Zipf formula}). For given $N$ users, we define $f_{\textup{hit}}$ as the hitting probability which denotes total popularity of contents cached in all users and can be defined as
\begin{equation}
\label{hitting probe}
f_{\textup{hit}} = \sum_{i=1}^{N} f(i).
\end{equation}
By using \textit{Thinning theorem} \cite{HaenggiBook}, in the following lemma, we obtain the expected number of nodes within a circle region which are considered as demanding users, and consequently can potentially operate in D2D communications. This means that the density of nodes which are considered as demanding users is $\mu \lambda$ as described in section \ref{System Model Section}.
\begin{lemma}
	\label{lemma N}
In a PPP network, the expected number of nodes denoted by $N$ within a ball of radius $R$, is $N =\mu \lambda \pi R^{2}$.
\end{lemma}
\begin{proof}
By taking expectation of the Poisson distributed random variable $N$, where the probability density function (PDF) of the Poisson distribution is $\frac{e^{-\mu\lambda A} \left(\mu\lambda A\right)^{N}}{N!}$, we can get $\mathbb{E}[N]=\mu \lambda A$, where $A$ is the area of the region and in our case it is $A=\pi R^{2}$.
\end{proof}
Since the number of nodes is an integer value, we round $N$ toward positive infinity in lemma \ref{lemma N}. Now, by considering the expected value of $N$, we can define the following propositions for a $\delta$-D2D network, where $\delta \in \{\textup{HD},\textup{FD}\}$.
\begin{proposition}
	\label{optimal caching proposition}
	In a wireless $\delta$-D2D enabled network with optimal caching policy, described in section \ref{System Model Section}, the probability that a user operates in $\delta$ mode, denoted by $\mathcal{P}_\delta$, can be defined as 
	
	\begin{equation}
	\label{HD/FD collaboration proposition formula}
	\mathcal{P}_{\delta} =  \sum_{i=1}^{N} \mathcal{Q}_\delta(i) \Lambda(i)  {f(i)},
	\end{equation}
	where, 
	\begin{equation}
	\Lambda(i) = \sum\limits_{n= 1}^{N - 1}{\left( {\begin{array}{*{20}{c}}
			{N - 1}\\
			n
			\end{array}} \right)}{\left( {f(i)} \right)^n}{\left( {1 - f(i)} \right)^{N - n - 1}},
	\end{equation}
	\begin{equation}
	\mathcal{Q}_{\textup{HD}}(i) = 1- \left( f_{\textup{hit}} - f(i)\right),
	\end{equation}
	\begin{equation}
	\mathcal{Q}_{\textup{FD}}(i) =  f_{\textup{hit}} - f(i).
	\end{equation}
\end{proposition}
\begin{proof}
	See Appendix \ref{proof for optimal caching proposition}. 
\end{proof}
\section{Performance Analysis}
\label{Performance Analysis Section}
First, we define the following indicator function $\mathbbm{1}_\delta$ to simplify the upcoming expressions, i.e., 
$
{\mathbbm{1}_{\delta}} = \left\{ \begin{array}{l}
\begin{array}{*{20}{c}}
1&{;\delta = \textup{FD}}
\end{array}\\
\begin{array}{*{20}{c}}
0&{;\delta = \textup{HD}}
\end{array}
\end{array} \right.$

\subsection{SINR Analysis}
\label{SINR subsection}
For the $\delta$-D2D inband overlay cellular system described in section \ref{System Model Section}, the experienced signal-to-interference-plus-noise ratio (SINR), denoted by $\Upsilon^{\delta}$, at receiver of interest located at the origin,  can be defined as
\begin{equation}
\label{SINR formula}
\Upsilon ^{\delta}  = \frac{{\rho _r }}{{{\sigma ^2} + \mathcal{I}_{d,\delta }  + {\mathbbm{1}_{\delta}}.(\beta {\rho _d})}},
\end{equation}
where $\rho _r = {\rho _d}h{\mathcal{R}_d^{ - {\alpha}}}$, and 
\begin{equation}
\label{Interference expression}
\mathcal{I}_{d,\delta } = \bm{\sum}\limits_{{z_i} \in {\bm{\Phi} _\delta }\backslash \{ {z_o}\} } {{\rho _d}{h_i}{{\left\| {{z_i}} \right\|}^{ - {\alpha}}}}.
\end{equation}
\subsection{Laplace Transform of the Interferences}
In this subsection, we aim to derive Laplace transform of the interference $\mathcal{I}_{d,\delta}$ in eq. (\ref{Interference expression}).
\begin{lemma}
	\label{Laplace Transforms}
	Laplace transform of the interference in eq. (\ref{Interference expression}) can be written as
	\begin{equation}
	\label{Laplace I_d}
	{\mathcal{L}_{\mathcal{I}_{d,\delta}}}(s) = \exp \left( {\frac{- 2{\pi ^2}}{\alpha }}{ \mu \mathcal{P}_\delta \lambda {{(s{p_d})}^{\frac{2}{\alpha }}}\csc(\frac{2}{\alpha }\pi )} \right),
	\end{equation}
	where $\csc(x)=\frac{1}{\sin(x)}$.
\end{lemma}
\begin{proof}
	See Appendix \ref{proof for Laplace Transform}.
\end{proof}

\subsection{Outage Probability}
\label{Outage Subsection}
\begin{theorem}
	\label{Outage Theorem}
	For an inband-overlay $\delta$-D2D network described in section \ref{System Model Section}, the outage probability on cellular and D2D links denoted by $\mathbb{Q}_o$ can be calculated as
	\begin{align}
	\mathbb{Q}_o(\lambda,\theta_d,\alpha) = & \notag \mathbb{P}\left(\Upsilon^{{\delta}}  \le \theta_d \right) \\=&
	1- \mathcal{J} ({\theta _d},{\rho _d},\beta ,{\alpha}){\mathcal{L}_{\mathcal{I}_{d,\delta}}}(\frac{\theta_d \mathcal{R}_d^{\alpha}}{\rho_d}).
	\end{align}
		where, $\mathcal{J} ({\theta _d},{\rho _d},\beta ,{\alpha}) =\exp \left( {\frac{{ - {\theta _d}\mathcal{R}_d^{{\alpha}}({\sigma ^2} + {\mathbbm{1}_{\delta}}(\beta {\rho _d}))}}{{{\rho _d}}}} \right)$,

\end{theorem}
\begin{proof}
	See Appendix \ref{proof for Outage Theorem}.
\end{proof}

\subsection{Spectral Efficiency}
\label{Spectral Efficiency subsection}


\begin{theorem}
	\label{spectral efficiency theorem}
	In the $\delta$-D2D enabled cellular network, the spectral efficiency of a D2D link denoted by $\mathcal{S}_\delta$, can be calculated as
	\begin{align}
	\label{SE formula}
	\mathcal{S}_\delta = \frac{\kappa}{\ln(2)}\int_0^\infty {\frac{  e^{\frac{-\left(\sigma^2+\mathbbm{1}_{\delta}(\beta \rho_d)\right)\mathcal{R}_d^{\alpha}}{\rho_d}\tau}}{1+\tau}}  { {\mathcal{L}_{\mathcal{I}_{d,\delta }}}\left( {\frac{\tau \mathcal{R}_{d}^{\alpha}}{{{{\rho_d}}}}} \right)} d\tau.
	\end{align}
	where, 
$
\kappa  = \left\{ \begin{array}{l}
\begin{array}{*{20}{c}}
1&{;\delta  = HD}
\end{array}\\
\begin{array}{*{20}{c}}
2&{;\delta  = FD}
\end{array}
\end{array} \right.$
\end{theorem}
\begin{proof}
	See Appendix \ref{proof for SE theorem}.
\end{proof}

\subsection{Special Cases}
\subsubsection{$\alpha=4$}
\begin{corollary}
	\label{Special Case by alpha on Laplace Transforms}
	When $\alpha=4$, the Laplace transform can be simplified as
	\begin{equation}
	\label{Special Case by alpha}
	{\mathcal{L}_{\mathcal{I}_{d,\delta}}}(s) =  \exp \left( {\frac{- {\pi ^2}}{2 }}{ \mu{\mathcal{P}_\delta}\lambda {\sqrt{s{p_d}}}} \right).
	\end{equation}
\end{corollary}
\begin{proof}
	By substituting $\alpha=4$ in eq. (\ref{Laplace I_d}), we can the expression in eq. (\ref{Special Case by alpha}).
\end{proof}

	
The interference power dominates background noise. Hence, background noise can often be neglected, i.e., $\sigma^2=0$. 
\subsubsection{Interference limited regime, ($\alpha=4$) and Perfectly Canceled SI}
\begin{corollary}
	If $\sigma^2=0$, $\alpha=4$ and $\beta=0$, then the spectral efficiency in eq. (\ref{SE formula}) can be directly solved, and the final expression will be 
	\begin{equation}
	\mathcal{S}_\delta=\frac{\kappa}{\ln(2)} \left[ \pi-2Si(T)  \right]\sin(T)-2ci(T)\cos(T),
	\end{equation}
where, $T= -\frac{\pi^2}{2} \mu\mathcal{P}_\delta \lambda \mathcal{R}_d^2$, and $Si(x)=\int_{0}^{x}\frac{\sin(t)}{t}dt$, $ci(x) = - \int_{x}^{\infty} \frac{\cos(t)}{t} dt$, respectively, are standard Sine and Cosine special integral functions, respectively.
\end{corollary}
\begin{proof}
	By directly using the formula \cite{IntegralBook} $\int_{0}^{\infty} \frac{e^{-\sqrt{u}}du}{u+\psi}=\left[\pi-2Si(\sqrt{\psi})\right]\sin(\sqrt{\psi})-2ci(\sqrt{\psi})\cos(\sqrt{\psi})$, we get the final expression. 
\end{proof}
\section{Results and Discussions}
\label{Results and Discussions Section}
Fig. \ref{Probe versus lambda zipf Fig} demonstrates collaboration probabilities versus $\lambda$ for HD and FD modes, for different Zipf exponents. For lower densities, users' demands mostly can be captured via HD collaboration. However, When the density increases, FD collaboration outperforms HD one, because, for any arbitrary node, the probability of finding users' demand within vicinity increases. In other words, due to increasing network, cache hit probability increases which is directly relates to successful cache hits. Zipf exponent in Fig. \ref{Probe versus lambda zipf Fig} plays crucial role, so that the high redundancy in the contents popularity, i.e., most popular contents capture majority of the interests, leads to multiple requests to the same content. By increasing network density, due to the increasing the number of interfering transmitters, for a given D2D link SINR threshold ($\theta_d$), the outage probability for both HD and FD modes increases (Fig. \ref{Outage vs lambda mu Fig}). The impact of the parameter $\mu$ is clear since it accounts for the density of demanding users which leads to increasing transmitters and consequently increasing outage probability in both HD and FD modes. The main difference between HD and FD mode is because of SI cancellation factor $\beta$ which have major impact on FD performance. We demonstrate the impact of the parameter $\beta$ in Fig. \ref{Outage vs theta beta Fig}. Typical values for efficient digital and analog SI cancellation in FD radios reported by \cite{full-duplex-advances1} which is around $10^{-5}$ and $10^{-6}$, i.e., $-50 \textup{ dB}$ and $-60\textup{ dB}$. AS can be seen from Fig. \ref{Outage vs theta beta Fig}, better SI cancellation factor (lower $\beta$) can increase FD radios performance and consequently decrease outage probability of FD mode. Outage probability's behavior for the system threshold $\theta_d$ is clear since increasing $\theta_d$ means guaranteeing D2D link quality for a successful decoding and demodulation at receivers. The impact of parameter $\mu$ in outage probability by considering different system SINR thresholds in Fig. \ref{Outage vs theta mu Fig} is similar to Fig. \ref{Outage vs lambda mu Fig}. Fig. The spectral efficiency of the D2D link is demonstrated in Fig. \ref{SE versus lambda beta Fig} for different values of parameter $\beta$. As can be seen from this figure, FD mode's performance converges to double against its HD counterpart, i.e., better SI cancellation factor (lower $\beta$) can theoretically double the specral efficiency of an FD link. In higher densities, due to higher interference, spectral efficiency for both HD and FD modes decreases.

\begin{figure}[h]
	\centering
	\includegraphics[width=0.5 \textwidth]{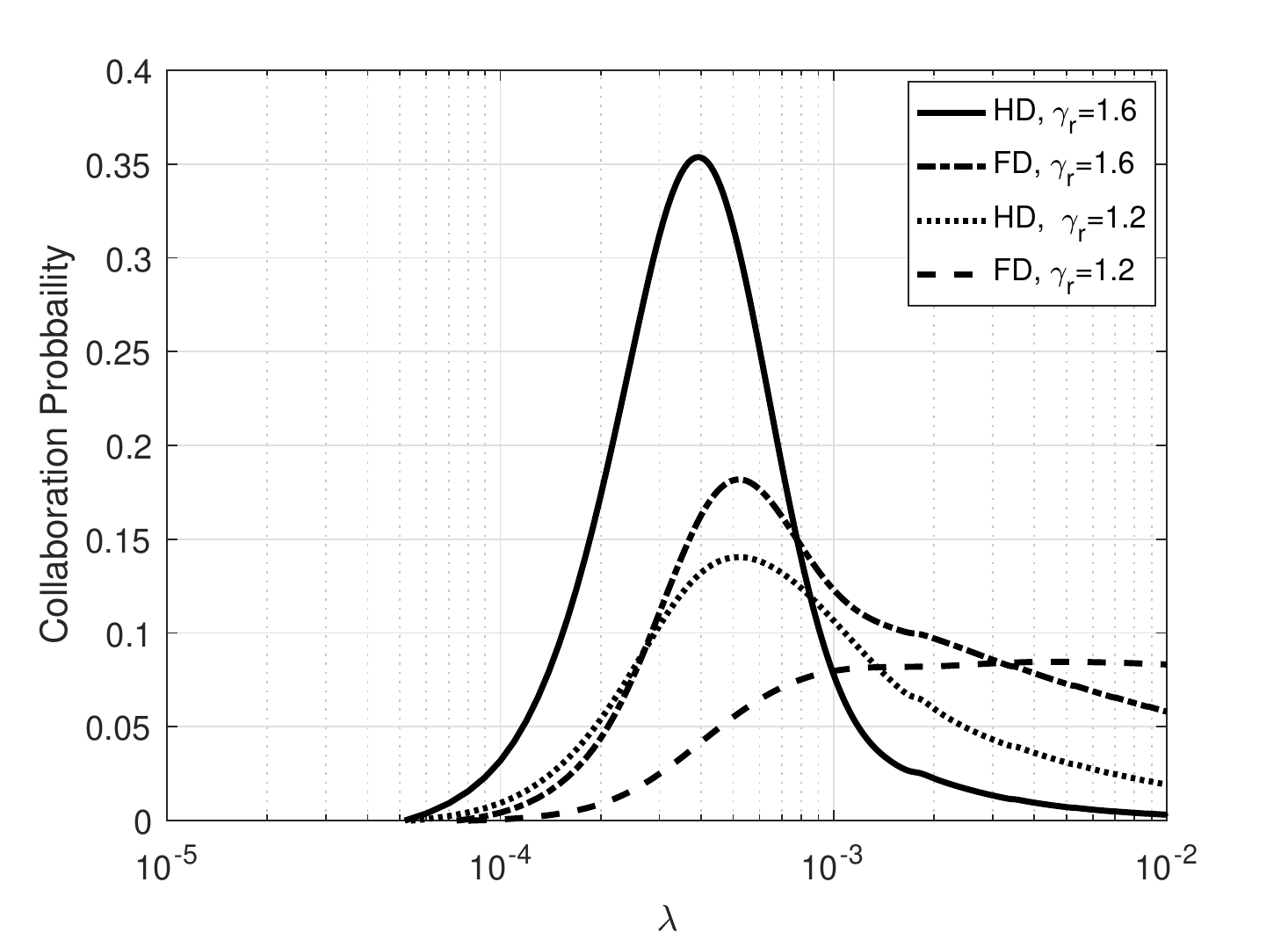}
	\caption{Collaboration Probability versus $\lambda$ for $\mu=0.3$.}
	\label{Probe versus lambda zipf Fig}
\end{figure}
\begin{figure}[h]
	\centering
	\includegraphics[width=0.5 \textwidth]{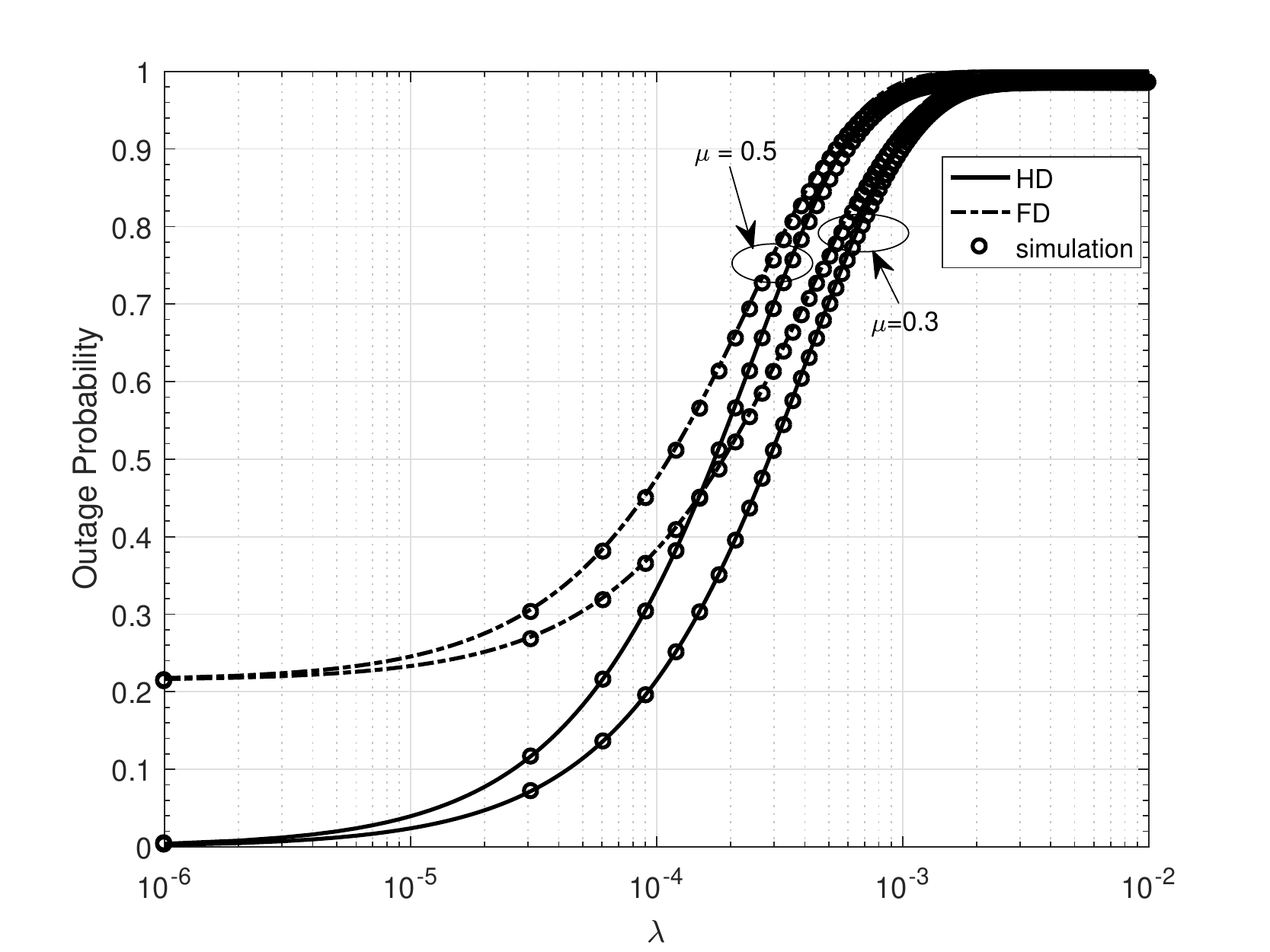}
	\caption{Outage probability versus $\lambda$ for $\gamma_r=1.2$, $\beta=10^{-5}$, $\theta_d=10 \textup{ dBm}$.}
	\label{Outage vs lambda mu Fig}
\end{figure}
\begin{figure}[h]
	\centering
	\includegraphics[width=0.5 \textwidth]{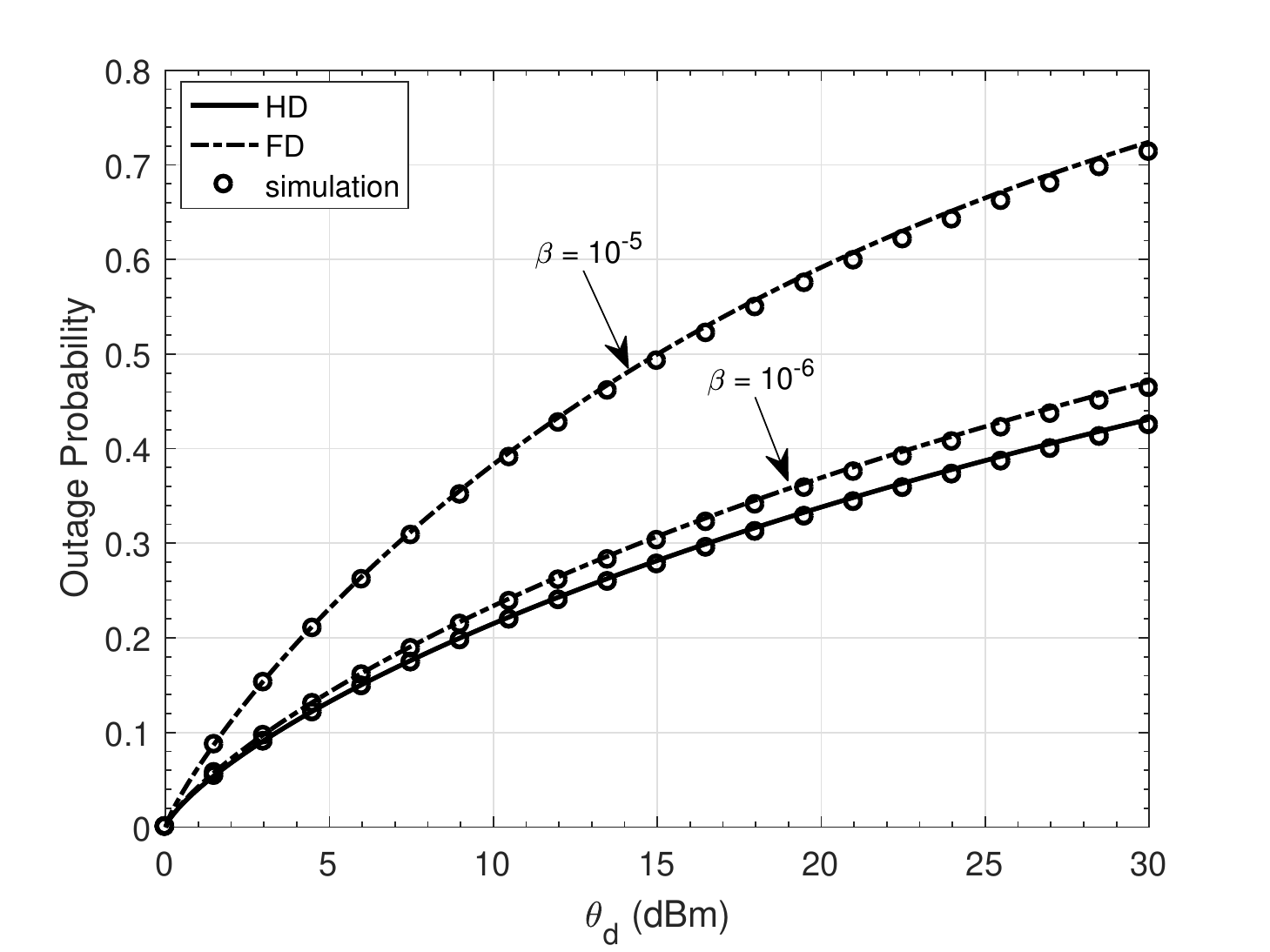}
	\caption{Outage probability versus $\theta_d$ for $\mu=0.3$, $\gamma_r=1.2$, $\lambda=10^{-3}$.}
	\label{Outage vs theta beta Fig}
\end{figure}
\begin{figure}[h]
	\centering
	\includegraphics[width=0.5 \textwidth]{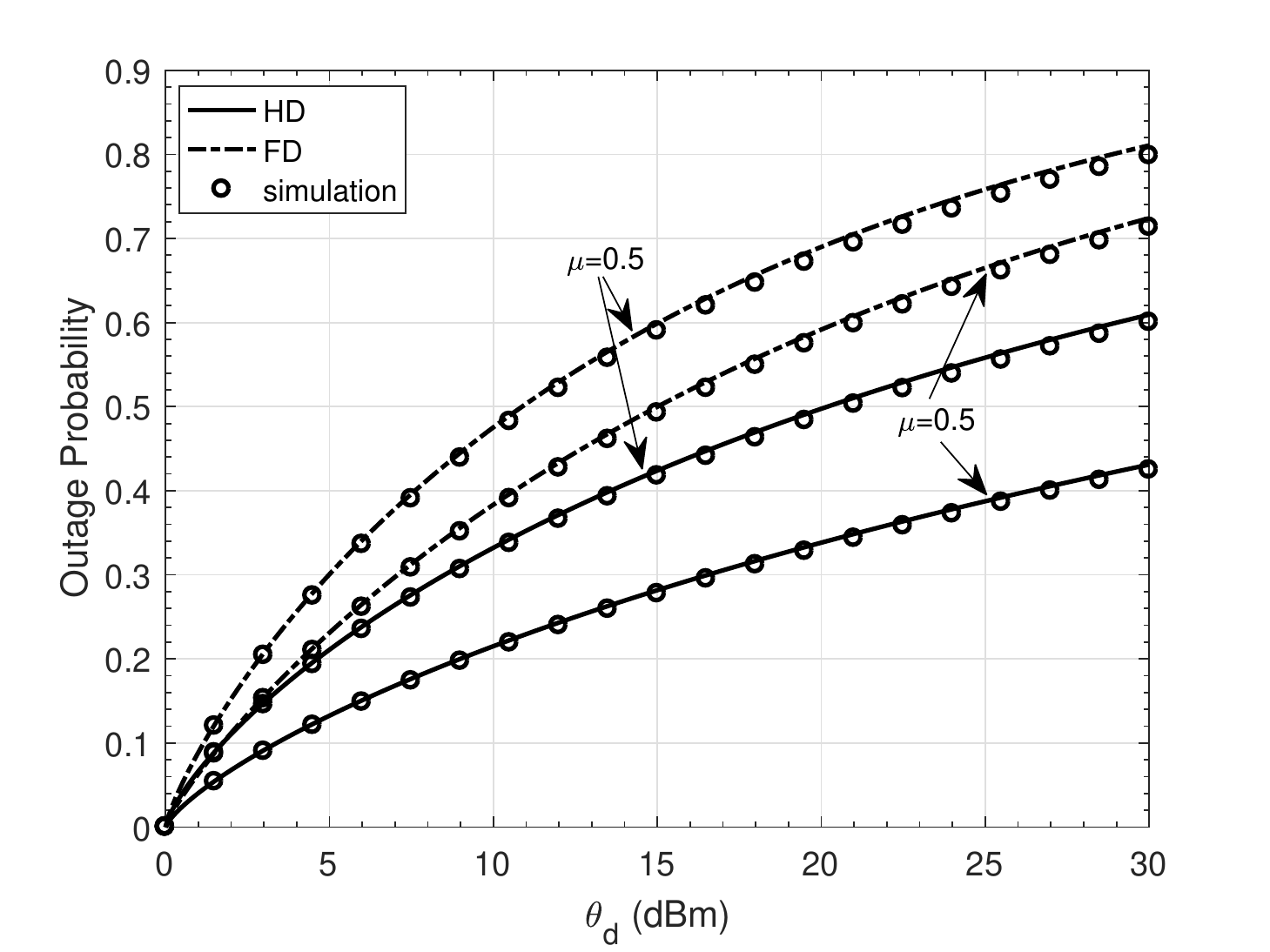}
	\caption{Outage probability versus $\theta_d$ for $\gamma_r=1.2$, $\beta=10^{-5}$, and $\lambda=10^{-3}$.}
	\label{Outage vs theta mu Fig}
\end{figure}
\begin{figure}[h]
	\centering
	\includegraphics[width=0.5 \textwidth]{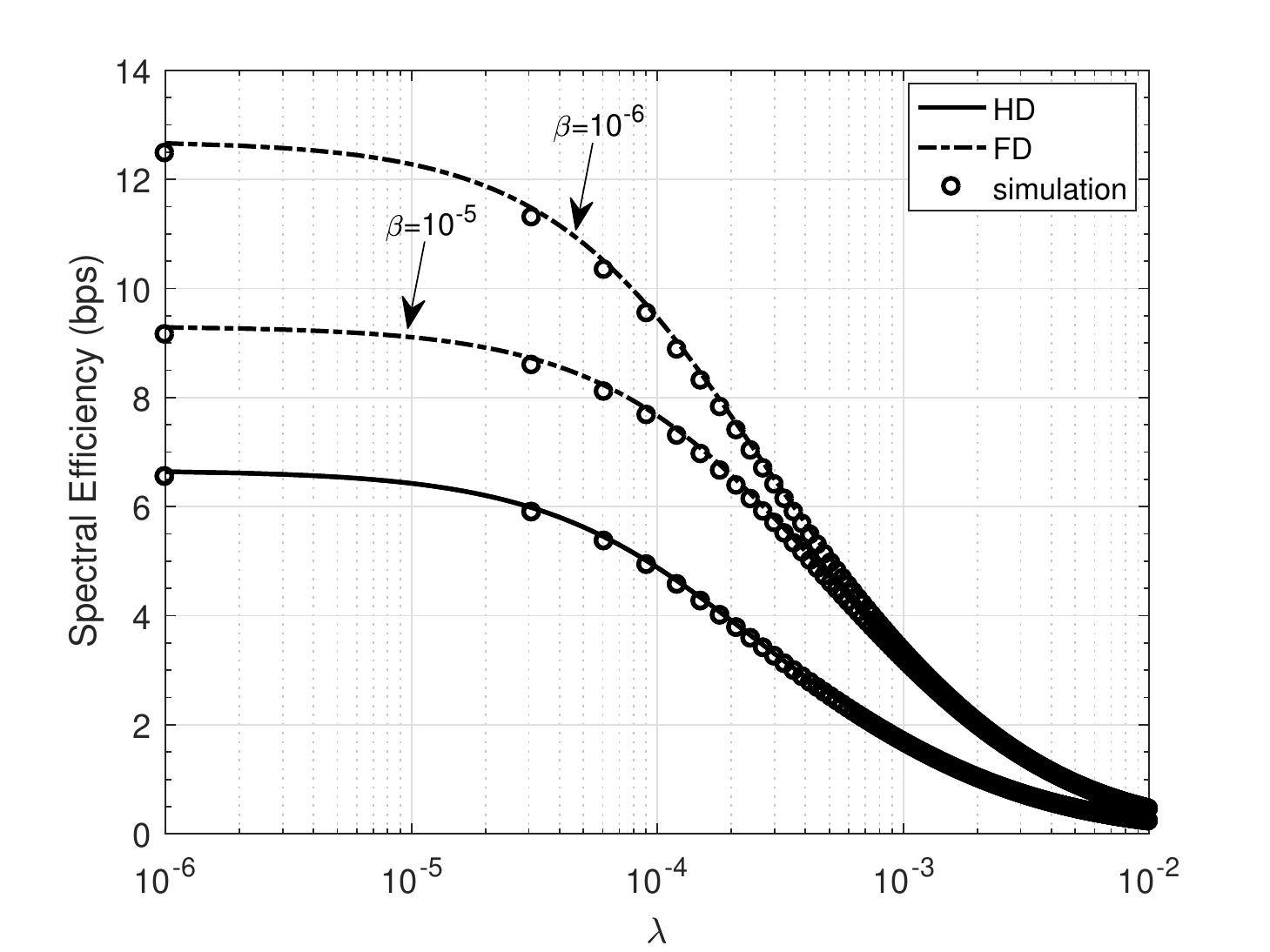}
	\caption{Outage probability versus $\lambda$ for $\mu=0.3$, $\gamma_r=1.2$.}
	\label{SE versus lambda beta Fig}
\end{figure}
\section{Conclusion}
\label{Conclusion Section}
In this paper, we analyzed the performance of the FD-enabled D2D network for the wireless video distribution. While FD radios can leads to more satisfaction in users' demand, analytic and simulation results confirms that there is a tradeoff between outage probability and the spectral efficiency; in which that FD radios suffers from the SI and hence leads to higher outage probability, while it can potentially double the spectral efficiency if a better SI cancellation factor is employed. 

\appendices
\section{Proof of Proposition \ref{optimal caching proposition}}

\label{proof for optimal caching proposition}

We define $\mathcal{P}_{\delta|i}$ which is conditioned on cached contents on user $u_i$, as the probability that user $u_i$ can potentially operate in $\delta$ mode. We also define $\mathcal{Q}_{\delta}(i)$ as the probability that $u_i$'s request cannot be served ($\delta=\textup{HD}$) or can be served ($\delta=\textup{FD}$), and $\Lambda(i)$ as the probability that $u_i$ can serve for at least one user's demand. Since the requests of contents are identically and independent distributed (i.i.d) at each user, we can say
\begin{equation}
\label{P delta conditioned on i}
\mathcal{P}_{\delta|i}=\mathcal{Q}_{\delta}(i)\Lambda(i).
\end{equation}
From the explanations in section \ref{System Model Section}, we can observe that 
\begin{equation}
\label{Q HD (i)}
\mathcal{Q}_{\textup{HD}}(i)=1-\left( f_{\textup{hit}}-f(i)\right),
\end{equation}
\begin{equation}
\label{Q FD (i)}
\mathcal{Q}_{\textup{FD}}(i)= f_{\textup{hit}}-f(i).
\end{equation}
Now, given $N$ users and assuming user $u_i$ caches contents with popularity $f(i)$ as in eq. (\ref{popularity opt u_i}), and excluding user $u_i$ from $N$ users, the number of requests for cached contents at user $u_i$ is a random variable denoted by $n$ and follows binomial distribution with parameters $N-1$ and $f(i)$, i.e., $n\sim \mathcal{B}(N-1,f(i))$. Since, $u_i$ can serve for multiple requests at the same time, it can be easily shown that 
\begin{equation}
\label{Lambda (i)}
\Lambda(i) = \sum\limits_{n= 1}^{N - 1}{\left( {\begin{array}{*{20}{c}}
		{N - 1}\\
		n
		\end{array}} \right)}{\left( {f(i)} \right)^n}{\left( {1 - f(i)} \right)^{N - n - 1}}.
\end{equation}
By substituting eqs. (\ref{Q HD (i)}, \ref{Q FD (i)} and \ref{Lambda (i)}) in eq. (\ref{P delta conditioned on i}), we obtain final expression for $\mathcal{P}_{\delta|i}$ in eq. (\ref{P delta conditioned on i}). Given $N$ users, the probability of that an arbitrary user among $N$ users, can operate in $\delta$-D2D mode, denoted by $\mathcal{P}_{\delta|N}$, can be obtained by taking expectation over all possible values for $i$. Therefore we have $\mathcal{P}_{\delta|N}=\sum_{i=1}^{N} \mathcal{P}_{\delta|i}f(i)$, which completes the proof. 

\section{Proof of Lemma \ref{Laplace Transforms}}
\label{proof for Laplace Transform}
\begin{align}
{\mathcal{L}_{\mathcal{I}_{d, \delta}}}(s) = & \notag {\mathbb{E}_{\mathcal{I}_{d, \delta}}}[{e^{-s\mathcal{I}_{d, \delta}}}] \\=& \notag {\mathbb{E}_{\mathcal{I}_{d, \delta}}}\left[ {\exp \left( {\bm{\sum}\limits_{{z_i} \in {\Phi _\delta}} {{-s\rho _d}{h_i}{{\left\| {{z_i}} \right\|}^{ - {\alpha}}}} } \right)} \right] \\= & \notag {\mathbb{E}_{{z_i},{h_i}}}\left[ {\bm{\prod}\limits_{{z_i} \in {\Phi _\delta}} {\exp \left( {{-s\rho _d}{h_i}{{\left\| {{z_i}} \right\|}^{ - {\alpha}}}} \right)} } \right]\\=& \notag {\mathbb{E}_{{\Phi _\delta}}}\left[ {\bm{\prod}\limits_{{z_i} \in {\Phi _\delta}} {{\mathbb{E}_{{h_i}}}\left[ {\exp \left( { - s{\rho _d}{h_i}{{\left\| {{z_i}} \right\|}^{ - {\alpha }}}} \right)} \right]} } \right] \\\mathop  = \limits^{(a)}& \notag
{\mathbb{E}_{{\Phi _\delta}}}\left[ {\bm{\prod}\limits_{{z_i} \in {\Phi _\delta}} {\left( {\int_0^\infty  {{e^{ - s{\rho _d}{h_i}{{\left\| {{z_i}} \right\|}^{ - {\alpha}}}}}} {e^{{-h_i}}}d{h_i}} \right)} } \right] \\=& \notag
{\mathbb{E}_{{\Phi _\delta}}}\left[ {\bm{\prod}\limits_{{z_i} \in {\Phi _\delta}} {\frac{1}{{1 + s{\rho _d}{{\left\| {{z_i}} \right\|}^{ - {\alpha}}}}}} } \right] \\\mathop  = \limits^{(b)}& \notag
\exp \left( { - {\mu\mathcal{P}_\delta\lambda}\int_0^\infty  {\int_0^{2\pi } {\left[ {1 - \frac{1}{{1 + s{\rho _d}{r^{ - {\alpha}}}}}} \right]rdrd\phi } } } \right) \\\mathop  = \limits^{(c)}& \notag 
\exp \left( { - 2\pi \mu\mathcal{P}_\delta\lambda \int_0^\infty  {\left( {\frac{1}{{1 + {{(s{\rho _d})}^{ - 1}}{r^{{\alpha}}}}}} \right)rdr} } \right) \\\mathop  = \limits^{(d)}& \notag
\exp \left( {\frac{{ - 2\pi \mu\mathcal{P}_\delta\lambda {{(s{\rho_d})}^{\frac{2}{\alpha }}}}}{\alpha}\underbrace {\int_0^\infty  {\frac{{{u^{\frac{2}{\alpha} - 1}}du}}{{1 + u}}} }_{\Theta}} \right) \\\mathop  = \limits^{(e)}&
\exp \left( {\frac{- 2{\pi ^2}}{\alpha}{ \mu\mathcal{P}_\delta\lambda {{(s{\rho_d})}^{\frac{2}{\alpha }}}\csc(\frac{2}{\alpha}\pi )}} \right),
\end{align}
where (a) follows Rayleigh fading channel model which is exponential distribution with unit mean, i.e., $h \sim \exp (1)$, ($b$) follows using probability generating functional (PGFL) of a PPP and ($c$) follows the density of process $\Phi_\delta$ which is defined in section \ref{System Model Section}, $(d)$ follows substituting variable $u=(s\rho_d)^{-1}r^{\alpha}$, and $(e)$ is direct solution for integral $\Theta$ from the table \cite[3.241 eq. (2)]{IntegralBook} , i.e., $\int_0^\infty  {\frac{{{u^{\psi-1}}du}}{{1 + u^{\zeta}}}}  = \frac{\pi}{\zeta} \csc\left( {\frac{\psi \pi}{{{\zeta}}}} \right)$ Re$\{\zeta\}>$ Re$\{\psi\}>0$.

\section{Proof of Theorem \ref{Outage Theorem}}
\label{proof for Outage Theorem}
\begin{align}
\mathbb{P}\left(\Upsilon^{{\delta}} \le \theta_d \right)=  \notag &
\mathbb{P}\left(\frac{{\rho_r }}{{{\sigma ^2} + \mathcal{I}_{d,\delta }+\mathbbm{1}_\delta(\beta\rho_d)}}  \le \theta_d \right) \\ = &  \notag
\mathbb{P}\left( {{\rho _d}h\mathcal{R}_d^{  {\alpha}} \le {\theta _d}\left( {{\sigma ^2} + \mathcal{I}_{d,\delta }+\mathbbm{1}_\delta(\beta\rho_d)} \right)} \right) \\ = & \notag
\mathbb{P}\left( {h \le \frac{{{\theta _d}\mathcal{R}_d^{  {\alpha}}}}{{{\rho _d}}}\left( {{\sigma ^2} + \mathcal{I}_{d,\delta } +\mathbbm{1}_\delta(\beta\rho_d)} \right)} \right) \\\mathop  = \limits^{(a)}& \notag
1 - {e^{\frac{{ - {\theta _d}\mathcal{R}_d^{  {\alpha}}}}{{{\rho _d}}}\left( {{\sigma ^2} + \mathcal{I}_{d,\delta }+\mathbbm{1}_\delta(\beta\rho_d)} \right)}}\\=&  \notag
1 - {e^{\frac{{ - {\theta _d}\mathcal{R}_d^{  {\alpha}}\left({\sigma ^2}+\mathbbm{1}_\delta(\beta\rho_d)\right)
		}}{{{\rho _d}}}}} 
{\mathcal{L}_{\mathcal{I}_{d,\delta }}}\left( {\frac{{  {\theta _d}\mathcal{R}_d^{  {\alpha}}}}{{{\rho _d}}}} \right),
\end{align}
where (a) follows Rayleigh fading channel model which is exponential distribution with unit mean, i.e., $h \sim \exp (1)$.

\section{Proof of Theorem \ref{spectral efficiency theorem}}
\label{proof for SE theorem}
It is clear that inband FD mode provides two simultaneously links at the same time/frequency, while HD mode provides single link between a D2D pair. We assume that channels in both links of the FD mode are reciprocal, hence we denote $\kappa$ as the number of links for each mode as in eq. (\ref{SE formula}). Now, according to the Shannon-Hartley theorem, we have
\begin{align}
\mathcal{S}_\delta = & \notag
\kappa\mathbb{E} \left[ \log_2 \left( 1+ \Upsilon^{\delta} \right) \right] \\=& \notag
\kappa\int_{0}^{\infty} \mathbb{P}\left( \log_2 \left( 1+ \Upsilon^{\delta}\right) \ge t \right)dt \\=& \notag
\kappa\int_{0}^{\infty} \mathbb{P}\left( \Upsilon^{\delta} \ge 2^t -1 \right)dt \\=& \notag
\kappa\int_{0}^{\infty} \mathbb{P}\left( {\frac{{\rho _r}}{{{\sigma ^2} +  \mathcal{I}_{d,\delta }  + {\mathbbm{1}_{\delta}}(\beta {\rho _d})}}} \ge {2^t -1} \right)dt \\\mathop  = \limits^{(a)}& \notag
\kappa\int_{0}^{\infty} \mathbb{P}\left( h \ge  \frac{2^t -1}{\rho_d\mathcal{R}_d^{-\alpha}} \left(  {{{\sigma ^2} + \mathcal{I}_{d,\delta}  + {\mathbbm{1}_{\delta}}(\beta {\rho _d})}} \right) \right)dt \\\mathop  = \limits^{(b)}& \notag
\kappa\int_{0}^{\infty} e^{-\frac{(2^t-1)(\sigma^2+\mathbbm{1}_\delta (\beta\rho_d) )}{\rho_d\mathcal{R}_d^{-\alpha}}} e^{-\frac{2^t-1}{\rho_d\mathcal{R}_d^{-\alpha}}\left( \mathcal{I}_{d,\delta}\right)}dt \\=& \notag 
\kappa\int_{0}^{\infty} e^{-\frac{(2^t-1)(\sigma^2+\mathbbm{1}_\delta (\beta\rho_d) )}{\rho_d \mathcal{R}_d^{-\alpha}}} \mathcal{L}_{\mathcal{I}_{d,\delta}}\left( \frac{2^t-1}{\rho_d \mathcal{R}_d^{-\alpha}} \right)dt \\\mathop  = \limits^{(c)}&
\frac{\kappa}{\ln(2)}\int_0^\infty {\frac{  e^{\frac{-\left(\sigma^2+ \mathbbm{1}_{\delta}(\beta \rho_d)\right)\mathcal{R}_d^{\alpha}}{\rho_d}\tau}}{1+\tau}}  {{\mathcal{L}_{\mathcal{I}_{d,\delta }}}\left( {\frac{\tau \mathcal{R}_d^{\alpha}}{{{\rho_d}}}} \right)} d\tau,
\end{align}
where $(a)$ follows the substitution $\rho_r=\rho_d \mathcal{R}_d^{-\alpha}$ with simple manipulations, $(b)$ follows Rayleigh fading channel model which is modeled with the random variable $h$ with the exponential distribution of unit mean, and finally by using the substitution $(2^t-1) \to \tau$ in ($c$), and simple manipulations, we get the final results which completes the proof.

\bibliographystyle{IEEEtran}


\end{document}